\documentclass[runningheads,a4paper]{llncs}

\def\ifundefined#1{\expandafter\ifx\csname#1\endcsname\relax}

\usepackage{amssymb}
\usepackage{graphicx}
\usepackage{amsmath}
\usepackage[retainorgcmds]{IEEEtrantools}

\usepackage{color}
\usepackage{array}
\usepackage{longtable}
\usepackage{calc}
\usepackage{multirow}
\usepackage{hhline}
\usepackage{ifthen}

\usepackage{url}
\newcommand{\keywords}[1]{\par\addvspace\baselineskip
\noindent\keywordname\enspace\ignorespaces#1}

\begin{document}

\mainmatter  

\title{A Proof on Asymptotics of Wavelet Variance of a Long Memory Process by Using Taylor Expansion}

\titlerunning{Technical Report 2010: Leibniz Institute for Neurobiology}

\author{Wonsang You\inst{1} \and Wojciech Kordecki\inst{2}}

\authorrunning{TR10014
: Asymptotics of Wavelet Variance of a Long Memory Process}

\tocauthor{Wonsang You (Leibniz Institute for Neurobiology),
Wojciech Kordecki (Wroc{\l}aw University of Technology)}
\institute{Special Lab Non-Invasive Brain Imaging,\\
Leibniz Institute for Neurobiology, Magdeburg, Germany\\
\email{you@ifn-magdeburg.de}
\and
Institute of Biomedical Engineering and Instrumentation,\\
Wroclaw University of Technology, Wroclaw, Poland\\
\email{wojciech.kordecki@pwr.wroc.pl}}

\maketitle

\begin{abstract}
A long memory process has self-similarity or scale-invariant properties in low frequencies. We prove that the log of the scale-dependent wavelet variance for a long memory process is asymptotically proportional to scales by using the Taylor expansion of wavelet variances. 
\end{abstract}

\section{Introduction}
Many natural phenomena have shown self-similarity or fractal properties in diverse areas. It is well known that time series in neuroimaging such as electroencephalogram, magnetoencephalogram, and functional MRI have fractal properties \cite{bullmore04,maxim05}. Most of fractal time series can be well modeled as long memory processes.

A long memory process generally has scale-invariant properties in low frequencies; for example, the scale-dependent wavelet variance has asymptotically power law relationship in low-frequency scales, which implies its fractal property \cite{moulines07}. Such a property also has been demonstrated by Achard \textit{et al.} \cite{achard08} by exploiting the Taylor expansion of wavelet variances; regrettably, they did not provide the detailed proof. Here, we show the detailed procedure of Achard \textit{et al.}'s proof.

\section{Wavelet Variance of Long Memory Process}

A long memory process can be defined as follows as suggested by Moulines \textit{et al.} \cite{moulines07}; that is, a real-valued discrete process $ \mathbf{X}:=\left \{ X(t) \right \}_{t=1,\cdots,N} $ is said to have long memory if it is covariance stationary and have the spectral density
\begin{equation}
S_{\mathbf{X}}(f)=\left | 1-e^{-if} \right |^{-2d}S^{*}_{\mathbf{X}}(f) \label{stationarylm}
\end{equation}
for $ 0<d<1/2 $ where $ S_{\mathbf{X}}^{*}\left ( f \right ) $ is a non-negative symmetric function and is bounded on $ \left ( -\pi, \pi \right ) $.
\begin{theorem} \label{theorem:wv}
Let $W_{j,k} $ is the wavelet coefficient of the time series $ \mathbf{X} $ at the $j$-th scale and the $k$-th location, and let $\mathcal{H}_{j}(f) $ be the squared gain function of the wavelet filter such that
\begin{equation}
\mathcal{H}_{j}(f)\approx \left\{\begin{matrix}
2^j & \text{for } 2\pi/2^{j+1}\leq \left | f \right | \leq 2\pi/2^{j}\\ 
0 & \text{otherwise} 
\end{matrix}\right..
\end{equation} 
Then, $\gamma(j) :=\text{Var}\left ( W_{j,t} \right ) $ satisfies
\begin{equation}
\gamma(j) =2\pi\int_{-\pi}^{\pi}\mathcal{H}_{j}(f)S_{\mathbf{X}}(f)df.
\end{equation}
\end{theorem}
Achard \textit{et al.} \cite{achard08} made its Taylor expansion to estimate the long memory parameter of $\mathbf{X}$ as shown in the corollary \ref{corollary:awv}. We give the detailed proof of the corollary \ref{corollary:awv}. 
\begin{corollary} \label{corollary:awv}
Let $ S_{\mathbf{X}}^{*}\left( f \right)=1+\beta f^{2}+o\left ( f^{2} \right ) $. Then, 
\begin{equation}
\gamma(j)=2\pi 2^{j+1}\int_{2\pi /2^{j+1}}^{2\pi / 2^{j}}\text{Re}\left ( S(f) \right )\,df = \widetilde{K}2^{2dj}\left(1 + a_{2}\frac{1}{2^{2j}} + o\left( \frac{1}{2^{2j}} \right) \right)
\end{equation}
where
\begin{equation}\label{eq:tildeK}
\widetilde{K}=2\frac{ 1-1/2^{1-2d}}{1-2d}\left ( 2\pi \right )^{2-2d},
\end{equation}
\begin{equation}\label{eq:a2}
a_{2}=\left ( 2\pi \right )^{2}\left ( \frac{d}{12} +\beta\right )\frac{\left ( 1-1/2^{3-2d} \right )/\left ( 3-2d \right )}{\left ( 1-1/2^{1-2d} \right )/\left ( 1-2d \right )}\,.
\end{equation}
\end{corollary}
\begin{proof}
It is well known 
we have the following Taylor series
\begin{equation}\label{eq:est_sin2}
\sin^{-2d}\left ( f/2 \right )=\left ( f/2 \right )^{-2d}+\frac{2d}{6}\left ( \frac{f}{2} \right )^{2\left ( 1-d \right )}+o\left ( f^3 \right ).
\end{equation}
From \eqref{stationarylm} and \eqref{eq:est_sin2},
\begin{equation}\label{eq:est_S(f)}
S\left ( f \right )=2^{-2d}\left ( \left ( \frac{f}{2} \right )^{-2d}+\frac{2d}{6}\left ( \frac{f}{2} \right )^{2\left ( 1-d \right )} \right )\left ( 1+\beta f^{2}+o\left ( f^{2} \right ) \right ).
\end{equation}
Let us define four integrals:
\begin{equation}\label{eq:I_1}
I_1:=\int_{2\pi /2^{j+1}}^{2\pi/2^{j}}
\left ( \frac{f}{2} \right )^{-2d}df=
\frac{\pi\left({2}^{2d}-2 \right){2}^{j\left ( 2d-1 \right )}}{\left( 2d-1\right){\pi }^{2d}}:=A_{1}2^{j(2d-1)},
\end{equation}
\begin{equation}\label{eq:I_2}
I_2:=\int_{2\pi /2^{j+1}}^{2\pi/2^{j}}
\left ( \frac{f}{2} \right )^{-2d}f^2 \,df 
=\frac{{\pi }^{3}\left({2}^{2d}-8\right){2}^{j\left ( 2d-3 \right )}}{\left( 2d-3\right)\pi^{2d}},
\end{equation}
\begin{equation}\label{eq:I_3}
I_3:=\int_{2\pi /2^{j+1}}^{2\pi/2^{j}}
\left ( \frac{f}{2} \right )^{2\left ( 1-d \right )}df= 
\frac{{\pi}^3\left ( 2^{2d}-8 \right ) 2^{j\left ( 2d-3 \right )}}{4\left ( 2d-3 \right )\pi^{2d}},
\end{equation}
\begin{equation}\label{eq:I_4}
I_4:=\int_{2\pi /2^{j+1}}^{2\pi/2^{j}}
\left ( \frac{f}{2} \right )^{2\left ( 1-d \right )}f^2\,df=
\frac{{\pi}^{5}\left ( 2^{2d}-32 \right )2^{j\left ( 2d-5 \right )}}{4\left ( 2d-5 \right )\pi^{2d}}:=A_{4}2^{j(2d-5)}.
\end{equation} 
Since if $h\left ( \Delta f \right )\to 0$ for $\Delta f\to 0$ then $ \int_a^{a+\Delta f} h\left ( \Delta f \right )\,df=o( \left ( \Delta f \right )^2 ) $,
\begin{equation}\label{eq:I_5}
I_5:=\int_{2\pi /2^{j+1}}^{2\pi/2^{j}}
\left ( \frac{f}{2} \right )^{-2d}o(f^2)\,df=2^{2dj}o\left ( \frac{1}{2^{2j}} \right ),
\end{equation} 
\begin{equation}\label{eq:I_6}
I_6:=\int_{2\pi /2^{j+1}}^{2\pi/2^{j}}
\left ( \frac{f}{2} \right )^{2\left ( 1-d \right )}o(f^2)\,df=2^{2dj}o\left ( \frac{1}{2^{4j}} \right ),
\end{equation} 
From \eqref{eq:est_S(f)} and the integrals \eqref{eq:I_1}-\eqref{eq:I_6},
\begin{IEEEeqnarray*}{rCl}\label{eq:S(f)-int-sol}
\gamma(j)&=&2\pi 2^{j+1}\int_{2\pi /2^{j+1}}^{2\pi / 2^{j}}\text{Re}\left ( S(f) \right )\,df\\
&=&\pi 2^{j+2(1-d)}\left [ I_1 + \beta I_2 + \frac{2d}{6}I_3 + \frac{2d}{6}\beta I_4 + I_5 + \frac{2d}{6}I_6 \right ]\\
&=&\pi 2^{j+2(1-d)}\left [ I_1 + \beta I_2 + \frac{2d}{6}I_3 + \frac{2d}{6}\beta I_4 + 2^{2dj}o\left ( \frac{1}{2^{2j}} \right ) \right ]\\
\end{IEEEeqnarray*}
Note that $ I_{3}=I_{2}/4 $ and
\begin{equation}\label{eq:I1I2}
I_2=\pi^2\frac{2^{2d}-8}{2^{2d}-2}\,\frac{2d-1}{2d-3}\,\frac{1}{2^{2j}}I_1
=\left ( 2\pi \right )^2\frac{1-2^{2d-3}}{1-2^{2d-1}}\,\frac{1-2d}{3-2d}\,\frac{1}{2^{2j}}I_1=M_{12}2^{-2j}I_{1}.
\end{equation}
Thus,
\begin{IEEEeqnarray*}{rCl}\label{eq:S(f)-int-sol2}
\gamma(j)&=&\pi 2^{j+2(1-d)}\left [ \left ( 1+\frac{a_2}{2^{2j}} \right )I_1 + \frac{2d}{6}\beta I_4 +2^{2dj}o\left ( \frac{1}{2^{2j}} \right ) \right ]\\
&=&\pi 2^{j+2(1-d)}\left [ \left ( 1+\frac{a_2}{2^{2j}} \right )A_1 2^{j(2d-1)} + \frac{2d}{6}\beta A_4 2^{j(2d-5)} +2^{2dj}o\left ( \frac{1}{2^{2j}} \right ) \right ]\\
\end{IEEEeqnarray*}
where
\begin{IEEEeqnarray*}{rCl}\label{eq:S(f)-int-sol3}
a_2&:=&\left ( \beta + d/12 \right )M_{12}\\
&=&\left ( 2\pi \right )^{2}\left ( \frac{d}{12} +\beta\right )\frac{\left ( 1-1/2^{3-2d} \right )/\left ( 3-2d \right )}{\left ( 1-1/2^{1-2d} \right )/\left ( 1-2d \right )}.\\
\end{IEEEeqnarray*}
Then, we can have
\begin{IEEEeqnarray*}{rCl}\label{eq:S(f)-int-sol4}
\gamma(j)&=&\widetilde{K}2^{2dj}\left [ 1+a_2 \frac{1}{2^{2j}} + \frac{2d}{6} \frac{A_4}{A_1} \beta \frac{1}{2^{4j}} + \frac{1}{\widetilde{K}}o\left ( \frac{1}{2^{2j}} \right )\right ]\\
&=&\widetilde{K}2^{2dj}\left [ 1+a_2 \frac{1}{2^{2j}} + o\left ( \frac{1}{2^{2j}} \right )\right ]\\
\end{IEEEeqnarray*}
where
\begin{equation}
\widetilde{K}:=\pi A_1 2^{2(1-d)}=2\frac{ 1-1/2^{1-2d}}{1-2d}\left ( 2\pi \right )^{2-2d}.
\end{equation}
\end{proof}

In the corollary \eqref{corollary:awv}, as the scale $ j $ increases, the wavelet variance $ \gamma(j) $ converges; in other words,
\begin{equation}
\gamma(j) \to \widetilde{K}2^{2dj}.
\end{equation}
It implies that we can simply estimate the long memory parameter $ d $ by linear regression method if we have the estimates of wavelet variances.

\section{Discussion}
We verified Achard \textit{et al.}'s proof on asymptotic properties of wavelet variances of a long memory process. Unfortunately, Achard \textit{et al.}'s Taylor expansion of wavelet variance is not perfectly consistent with our result; while $ \widetilde{K} = 2\frac{ 1-1/2^{1-2d}}{1-2d}\left ( 2\pi \right )^{2-2d} $ in our result, they computed as $ \widetilde{K} = 2\frac{ 1-1/2^{1-2d}}{1-2d}\left ( 2\pi \right )^{1-2d} $.

On the other hand, the assumption on short memory such that $ S_{\mathbf{X}}^{*}\left ( f \right )=1+\beta f^{2}+o\left ( f^{2} \right ) $ loses generality of their proof since there exist more general classes of short memory; indeed, Moulines \textit{et al.} proved that the asymptotics of wavelet variances hold when the short memory $ S_{\mathbf{X}}^{*}\left ( f \right ) $ belongs to the function set $ \mathcal{H}\left ( \beta,L \right ) $ defined as the set of even non-negative functions $ g $ on $ \left [ -1/2, 1/2 \right ] $ such that
\begin{equation}
\left | g(f) - g(0) \right | \leq Lg(0)\left | 2 \pi f \right |^{\beta}. \label{gsd}
\end{equation}
Nevertheless, their method is relatively simpler than other mathematical proofs. Moreover, we will be able to apply their method to investigate the asymptotic properties of multivariate long memory processes as Achard \textit{et al.} already attempted for bivariate long memory processes \cite{achard08}.

\end{document}